\documentclass{article}
\usepackage[a4paper]{geometry}

\usepackage[utf8]{inputenc}
\usepackage[T1]{fontenc}
\usepackage{microtype}
\usepackage{flushend}

\usepackage{amsmath}
\usepackage{amssymb}
\usepackage{amsthm}
\usepackage{graphicx}
\usepackage{subcaption}
\usepackage{caption}
\usepackage{paralist}
\usepackage{todonotes}
\usepackage{xspace}

\newtheorem{theorem}{Theorem}
\newtheorem{lemma}{Lemma}
\newtheorem{corollary}{Corollary}

\newtheorem{definition}{Definition}

\newcommand{\Oh}{{\ensuremath{\mathcal{O}}}}

\newcommand{\Qn}{{\ensuremath{\textsf{qn}}}}

\newcommand{\fullBM}{$\Delta$-matched\xspace}

\title{Planar Graphs of Bounded Degree have\\Bounded Queue Number%
\thanks{%
Work partially supported by the DFG grant Ka812/17-1, by DAAD project~57419183, and by the MIUR grant 20174LF3T8 AHeAD: efficient Algorithms for HArnessing networked Data.}}

\author{
Michael~A.~Bekos\footnotemark[2]%
\and 
Henry~F\"orster\footnotemark[2]%
\and 
Martin~Gronemann\footnotemark[3]%
\and 
Tamara~Mchedlidze\footnotemark[4]%
\and 
Fabrizio~Montecchiani\footnotemark[5]%
\and 
Chrysanthi~Raftopoulou\footnotemark[6]%
\and 
Torsten~Ueckerdt\footnotemark[4]%
\\[0.2in]
\footnotemark[2]~~University of T{\"u}bingen, T{\"u}bingen, Germany \\\texttt{\small bekos@informatik.uni-tuebingen.de} and \texttt{\small foersth@informatik.uni-tuebingen.de}
\\[0.1in]
\footnotemark[3]~~University of Cologne, Cologne, Germany \\\texttt{\small gronemann@informatik.uni-koeln.de}
\\[0.1in]
\footnotemark[4]~~Karlsruhe Institute of Technology, Karlsruhe, Germany \\\texttt{\small mched@iti.uka.de} and \texttt{\small torsten.ueckerdt@kit.edu}
\\[0.1in]
\footnotemark[5]~~University of Perugia, Perugia, Italy\\\texttt{\small fabrizio.montecchiani@unipg.it}
\\[0.1in]
\footnotemark[6]~~National Technical University of Athens, Athens, Greece\\\texttt{\small crisraft@mail.ntua.gr}
}

\date{} 

\begin{document}

\maketitle

\begin{abstract}
A \emph{queue layout} of a graph consists of a \emph{linear order} of its vertices and a partition of its edges into \emph{queues}, so that no two independent edges of the same queue are nested. The \emph{queue number} of a graph is the minimum number of queues required by any of its queue layouts.
A long-standing conjecture by Heath, Leighton and Rosenberg states that the queue number of planar graphs is bounded. This conjecture has been partially settled in the positive for several subfamilies of planar graphs (most of which have bounded treewidth). In this paper, we make a further important step towards settling this conjecture. We prove that planar graphs of bounded degree (which may have unbounded treewidth) have bounded queue number. 

A notable implication of this result is that every planar graph of bounded degree admits a three-dimensional straight-line grid drawing in linear volume. Further implications are that  every planar graph of bounded degree has bounded track number, and that every $k$-planar graph (i.e., every graph that can be drawn in the plane with at most $k$ crossings per edge) of bounded degree has bounded queue number.  
\end{abstract}

\section{Introduction}
\label{sec:introduction}

Queue layouts of graphs form a well-known type of linear layout and play an important role in various fields, such as in sorting~\cite{T72}, scheduling~\cite{BCLR96}, VLSI circuit design~\cite{LR86}, matrix computations~\cite{Pem92} and graph drawing~\cite{DLM05,DMW05}; see also \cite{DW04,Pem92} for further applications. A \emph{queue layout} of a graph consists of a vertex ordering and a partition of its edges, so that no two independent edges in the same part, called \emph{queue}, are nested~\cite{HR92}; see Figure~\ref{fig:example} for an illustration. The \emph{queue number} $\Qn(G)$ of a graph $G$ is the minimum number of queues in any queue layout of $G$. Note that queue layouts form the ``dual'' concept of \emph{stack layouts}~\cite{Oll73,Yan89} (widely known as \emph{book embeddings}), in which two edges of the same stack are allowed to nest but not to cross. 

It is known that there exist non-planar graphs on $n$ vertices with $\Theta(n)$ queue number; for example, the queue number of the complete graph $K_n$ is $\lfloor n/2 \rfloor$~\cite{HR92}. Moreover, there exist graphs of bounded degree that may require arbitrarily many queues~\cite{HLR92,DBLP:journals/dmtcs/Wood08}.
Concerning sublinear upper bounds, graphs with $m$ edges have queue number $O(\sqrt{m})$~\cite{HLR92}, while graphs with $n$ vertices that belong to any minor-closed graph family have queue number $\log^{O(1)}n$~\cite{DMW17}. Bounded queue number is achieved by all graphs of bounded treewidth~\cite{DMW05}. In particular, a graph with treewidth $w$ has queue number $\Oh(2^w)$~\cite{Wie17}. Improved bounds (linear in the parameter) are known for graphs of bounded pathwidth~\cite{DMW05}, bounded track number~\cite{DW05}, bounded bandwidth~\cite{HLR92}, or bounded layered pathwidth~\cite{BDDEW18}; for a survey we refer the reader to~\cite{DMW17}.

\begin{figure*}[t]
	\centering
	\includegraphics[width=0.45\textwidth,page=3]{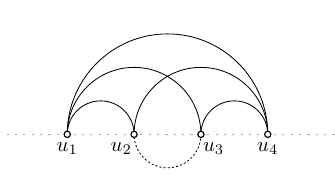}
    \caption{Illustration of a queue layout of the complete graph $K_4$ on $4$ vertices $\{u_1,\ldots,u_4\}$ with two queues (solid and dotted).}
	\label{fig:example}
\end{figure*}

A rich body of literature focuses on planar graphs. In fact, it is known that the graphs that admit $1$-queue layouts are the arched-level planar graphs~\cite{HR92}, which are planar graphs with at most $2n-3$ edges over $n$ vertices (note that testing whether a graph is arched-level planar is NP-complete~\cite{HLR92}). Trees are arched-level planar and therefore have queue number one~\cite{HR92}. Outerplanar graphs have queue number at most two~\cite{HLR92}, Halin graphs and series-parallel graphs have queue number at most three~\cite{Ganley95,RM95}, and planar $3$-trees have queue number at most five~\cite{ABG0P18}. However, it is still unknown whether the queue number of planar graphs is bounded (see also~\cite{openproblem}). In particular, back in 1992, Heath, Leighton and Rosenberg~\cite{HLR92} conjectured that every planar graph has bounded queue number. Notably, this conjecture has not been settled after almost three decades. The best-known upper bound is due to Dujmovi{\'c}~\cite{Duj15}, who showed that the queue number of planar graphs on $n$ vertices is $\Oh(\log n)$ improving upon a previous bound by Di Battista et al.~\cite{BFP13}; recently, Bannister et al.~\cite{BDDEW18} improved the constant factor in the result of Dujmovi{\'c}~\cite{Duj15}. On the other hand, the best-known lower bound is due to a family of planar $3$-trees that require four queues~\cite{ABG0P18}. 

It is worth noting that a positive answer to the conjecture by Heath, Leighton and Rosenberg~\cite{HLR92} would have several implications. First, every planar graph would admit a three-dimensional grid drawing in linear volume~\cite{DMW05}, which is a major open problem in graph drawing first posed by Felsner et al.~\cite{DBLP:journals/jgaa/FelsnerLW03} back in 2003.  Second, every planar graph would have bounded track number~\cite{DPW04}. Third, every bipartite planar graph would have a $2$-layer drawing and a corresponding edge-coloring with a bounded number of colors, in which no two edges of the same color cross~\cite{DW05}. 
Indeed, the above three problems are equivalent to queue layouts~\cite{DMW05,DPW04,DW05}, that is, for instance, if planar graphs have bounded track number, then they also have bounded queue number.
Finally, it is known that if the queue number of planar graphs is bounded, then the same holds for the broader family of $k$-planar graphs, i.e., of those graphs that can be drawn in the plane such that each edge is crossed at most $k$ times (for fixed values of $k$)~\cite{DW05}. 

\subsection{Our Contribution} In this paper, we make an important step towards settling the conjecture by Heath, Leighton and Rosenberg~\cite{HLR92} that planar graphs have bounded queue number. Namely, we prove that every planar graph of maximum degree $\Delta$ has queue number $\Oh(\Delta^{c})$, where $c$ is a small constant. This implies that the conjecture holds for planar graphs of bounded degree. More precisely, the main contribution of this paper is the following theorem.

\begin{theorem}\label{thm:main}
Every planar graph of maximum degree $\Delta$ has queue number at most $32(2\Delta-1)^{6}-1$.
\end{theorem}

The proof of Theorem~\ref{thm:main} is constructive and yields an algorithm that computes a queue layout of the input graph $G$ with at most $32(2\Delta-1)^{6}-1$ queues in polynomial time (see Theorem~\ref{thm:time}). 

Closely related to queue layouts are the so-called \emph{track layouts}. In a track layout, the vertices of a graph are partitioned into sequences, called \emph{tracks}, such that the vertices in each track form an independent set and the edges between each pair of tracks do not cross. The \emph{track number} of a graph $G$ is the minimum number of tracks in a track layout of $G$. Whether planar graphs have bounded track number is still an open question, with the current best upper bound being $\Oh(\log{n})$~\cite{BDDEW18}. Dujmovi\'c, Morin and Wood~\cite{DMW05} prove that every graph with queue-number $q$ and acyclic chromatic number $k$ has track-number at most $q(2k)^{q-1}$. This result combined with Theorem~\ref{thm:main} implies that planar graphs of bounded degree have bounded track number, since planar graphs have acyclic chromatic number at most five~\cite{Borodin79}. More precisely:

\begin{corollary}\label{cor:tracks}
Every planar graph of maximum degree $\Delta$ has track number $\Delta^6\,2^{\Oh(\Delta^6)}$.
\end{corollary}

Dujmovi\'c and Wood~\cite{DW04b} proved that every $n$-vertex $c$-colourable graph with track number $t$ has a three-dimensional straight-line grid drawing with bounding box $\Oh(c) \times \Oh(c^2t) \times \Oh(c^5n)$.  Conversely, if a graph has a three-dimensional straight-line grid drawing with bounding box $X \cdot Y \cdot Z$, then it has track number at most $2XY$~\cite{DBLP:conf/gd/DujmovicMW02}.
This result combined with Corollary~\ref{cor:tracks} yields the following corollary. We remark that the best upper bound currently known for the volume of three-dimensional grid drawings of planar graphs is $\Oh(n \log{n})$~\cite{BDDEW18}.

\begin{corollary}
Every planar graph of maximum degree $\Delta$ admits a straight-line drawing on a $\Oh(1) \times \Delta^6\,2^{\Oh(\Delta^6)} \times \Oh(n)$ three-dimensional grid.
\end{corollary}

Another related problem is the \emph{$2$-track thickness} problem, in which one seeks for a $2$-layer drawing of a bipartite graph and a corresponding coloring of its edges with as few colors as possible, such that no two edges of the same color cross; the minimum number of colors required in any such drawing is referred to as \emph{$2$-track thickness}. Again, it is not known whether the $2$-track thickness of bipartite planar graphs is bounded, with the best upper bound being $\Oh(\log{n})$~\cite{DW05}, which is due to a result by Dujmovi\'c and Wood~\cite{DW05} that relates the queue number of a bipartite planar graph and its $2$-track thickness. Using the same result, together with Theorem~\ref{thm:main}, we obtain that every bipartite planar graph of bounded degree has bounded $2$-track thickness. 

\begin{corollary}
Every bipartite planar graph of maximum degree $\Delta$ has $2$-track thickness $\Oh(\Delta^6)$.
\end{corollary}

A last corollary of Theorem~\ref{thm:main} stems from another result by Dujmovi\'c and Wood~\cite{DW05}, who prove that if the queue number of planar graphs is bounded, then the same holds for $k$-planar graphs when $k$ is a fixed constant. Recall that a graph is $k$-planar if it admits a drawing such that each edge is crossed at most $k$ times. The proof in~\cite{DW05} relies on a planarization technique, in which each crossing is replaced by a degree-$4$ vertex. If the input is a $k$-planar graph of bounded degree, the obtained planarization is also of bounded degree, and we obtain the following corollary. Note that, Dujmovi\'c and Frati~\cite{DBLP:journals/jgaa/DujmovicF18} proved that $k$-planar graphs have queue number $\Oh(\log{n})$.   

\begin{corollary}
For fixed $k \geq 0$, every $k$-planar graph of maximum degree $\Delta$ has queue number $\Delta^{\Oh(k)}$.
\end{corollary}

\subsection{Proof strategy} Our approach starts with a layering of the vertices of the input graph obtained from a breadth-first search (BFS) traversal. This layering serves as a basis for computing the linear order of the vertices in the queue layout. In particular, the vertices that belong to earlier layers in the BFS-tree $T$ will precede those belonging to subsequent layers of $T$. This is used to prove that all edges that belong to $T$ do not nest in such a linear order. The main challenge with this approach is to deal with edges that are not part of $T$. Such edges either connect vertices on the same layer or vertices lying on consecutive layers. While the second type of edges can be eliminated by subdividing them a constant number of times, the first type of edges may still result in arbitrarily large groups of pairwise nesting edges, commonly called \emph{rainbows}. To cope with this issue, we change the order of the vertices on each layer so as to eliminate all nestings between edges connecting vertices of the same layer. On the other hand, this will unavoidably introduce rainbows formed by edges of $T$. However, we reorder the vertices such that the maximum number of edges of $T$ in the same rainbow is bounded by a polynomial in $\Delta$. For ease of description, we~first discuss our approach in a special case, namely when the edges that do not belong to $T$ form a perfect matching on its leaves. Then we show how to use the solution of the special case for the general~case.

\subsection{Paper organization and overview} In Section~\ref{sec:preliminaries}, we introduce notation and definitions that are used throughout this paper. We also present results from the literature that we exploit in our algorithm. In Sections~\ref{sec:subcase} and~\ref{sec:general}, we describe an efficient algorithm that takes a planar graph of maximum degree $\Delta$ as input, and outputs a queue layout of this graph with at most $32(2\Delta-1)^{6}-1$ queues. More precisely, in Section~\ref{sec:subcase} we introduce a special subfamily of plane graphs with maximum degree $\Delta$, for which it is possible to compute a queue layout with at most $2\Delta-2$ queues. Then, in Section~\ref{sec:general}, we reduce the problem of finding a queue layout of a general planar graph of degree $\Delta$ to the already discussed special case. This reduction increases the number of queues used to $32(2\Delta-1)^{6}-1$. Section~\ref{sec:timeComplexity} discusses the time complexity of our algorithm.  Section~\ref{sec:conclusions} concludes the paper with open problems.

\subsection{Subsequent work} We conclude our introduction by summarizing recent developments on the long-standing conjecture by Heath, Leighton and Rosenberg~\cite{HLR92}, triggered by the first version of our paper~\cite{DBLP:journals/corr/abs-1811-00816}. In a follow-up arXiv paper~\cite{DBLP:journals/corr/abs-1901-05594}, Dujmovi\'c, Morin, and Wood improved the upper bound of Theorem~\ref{thm:main}, and they extended the result to graphs with bounded genus. More precisely, they proved that graphs with Euler genus $g$ and maximum degree $\Delta$ have queue number $\Oh(g+\Delta^2)$. This is obtained by refining the analysis of our algorithm to obtain a $\Oh(\Delta^2)$ upper bound, and then by applying our algorithm as a black-box for graphs with genus $g$. It is worth remarking that the analysis in~\cite{DBLP:journals/corr/abs-1901-05594} (and hence the resulting $\Oh(\Delta^2)$ upper bound) is best possible, up to a constant factor. Most importantly and very recently, Dujmovi\'c, Joret, Micek, Morin, Ueckerdt, and Wood~\cite{DBLP:journals/corr/abs-1904-04791,focs} settled the conjecture by Heath, Leighton and Rosenberg~\cite{HLR92} by proving that every planar graph has queue number bounded by a constant, irrespective of the maximum degree. The proof is based on new structural tools and the result generalises for graphs of bounded Euler genus and further for every proper minor-closed class of graphs.

\section{Preliminaries}
\label{sec:preliminaries}
In this section we introduce preliminary concepts and known results that we use throughout the paper. Note that we assume familiarity with basic graph theoretic concepts; see, e.g.,~\cite{DBLP:books/daglib/0030488}.

\subsection{Queue Layouts} Let $G=(V,E)$ be an $n$-vertex \emph{simple} connected undirected graph, that is, a graph with neither self-loops nor multi-edges. We denote an edge between vertices $u$ and $v$ by $(u,v)$. Let $\prec$ be a linear order of the vertices of $G$. Consider two \emph{independent} edges $(u_1,v_1)$ and $(u_2,v_2)$ that do not share a common endvertex. Up to a renaming of their endvertices, we may assume that $u_1 \prec v_1$ and $u_2 \prec v_2$. We say that $(u_1,v_1)$ \emph{nests} $(u_2,v_2)$ with respect to $\prec$ if and only if $u_1 \prec u_2 \prec v_2 \prec v_1$; see Figure~\ref{fig:nest}. Consider now $k$ edges $(u_1,v_1),\ldots,(u_k,v_k)$ that are pairwise independent, such that $u_i \prec v_i$ for each $i=1,\ldots,k$. \begin{figure*}[t]
	\centering
	\subcaptionbox{\centering \label{fig:nest}}{
	\includegraphics[width=0.22\textwidth,page=2]{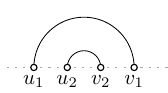}}
	\hfil
    \subcaptionbox{\centering \label{fig:necklace}}{
	\includegraphics[width=0.22\textwidth,page=4]{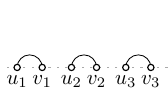}}
	\hfil
	\subcaptionbox{\centering \label{fig:rainbow}}{
	\includegraphics[width=0.22\textwidth,page=5]{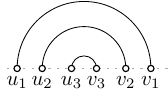}}
    \caption{Illustration of:  
    (a)~two nesting edges $(u_1,v_1)$ and $(u_2,v_2)$,
    (b)~a $3$-necklace, and 
    (c)~a $3$-rainbow.}
	\label{fig:linearLayouts}
\end{figure*}
If $u_1 \prec v_1 \prec u_2 \prec v_2 \prec \ldots \prec u_k \prec v_k$, then we say that edges $(u_1,v_1),\ldots,(u_k,v_k)$ form a \emph{$k$-necklace} in $\prec$; see Figure~\ref{fig:necklace}.
On the other hand, if $u_1 \prec \ldots \prec u_k \prec v_k \prec \ldots \prec v_1$, then we say that edges $(u_1,v_1),\ldots,(u_k,v_k)$ form a \emph{$k$-rainbow} in $\prec$; see Figure~\ref{fig:rainbow}. 

A preliminary result by Heath and Rosenberg~\cite{HR92} shows that a graph admits a queue layout with $k$ queues if and only if there exist a linear order of its vertices in which no $(k+1)$-rainbow is formed.
Another useful tool for analyzing the queue number of graphs is the fact that the queue number of a graph is bounded by the queue number of any of its subdivisions. More precisely: 

\begin{lemma}[Dujmovi\'c and Wood~\cite{DW05}]
Let $D$ be a subdivision of a graph $G$ obtained by subdividing each of the edges of $G$ at most $k$ times. If $\Qn(D) \le q$, then $\Qn(G) \le \frac{1}{2}(2q+2)^{2k}-1$.
\label{lem:subdivision}
\end{lemma}

\subsection{Planar Drawings} A \emph{drawing} of a graph $G$ is a mapping of the vertices of $G$ to distinct points of the plane, and of the edges of $G$ to Jordan arcs connecting their corresponding endvertices but not passing through any other vertex. A drawing is \emph{planar} if no two edges intersect, except possibly at a common endvertex. A graph is \emph{planar} if it admits a planar drawing. A planar drawing subdivides the plane into topologically connected regions, called \emph{faces}. The infinite region is called the~\emph{outer~face}.

\subsection{Ordered Concentric Representations} Central in our approach are also the so-called \emph{ordered concentric representations}, which were recently studied by Pupyrev~\cite{Pup17}. An ordered concentric representation of a planar graph $G$ is a planar drawing of $G$ where the vertices  are located at concentric circles $\mathcal{C}_0,\mathcal{C}_1,\ldots,\mathcal{C}_{h-1}$ with decreasing radii centered at a point $c$ of the plane, except for a single vertex $r$, called the \emph{center} of the representation, which is located at point $c$. Recall that the \emph{graph-theoretic distance} between two vertices of $G$ is the number of edges in a shortest path connecting them. All vertices on circle $\mathcal{C}_i$ have graph-theoretic distance $h-i$ from $r$, for $i=0,1,\ldots,h-1$. It follows that each edge of $G$ either has both its endvertices at the same circle (\emph{level edge}) or at two consecutive circles (\emph{binding edge}). 
\begin{figure}
\centering
	\subcaptionbox{\centering \label{fig:octahedron}}{
	\includegraphics[height=2.92cm,page=1]{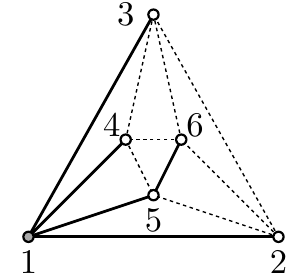}}
	\hfil
	\subcaptionbox{\centering \label{fig:representation}}{
	\includegraphics[height=2.92cm,page=2]{figs/representation}}
    \caption{%
    Illustration of:
    (a)~the octahedron graph with a BFS-tree (drawn bold) rooted at a specific vertex (colored gray), and 
    (b)~its corresponding ordered concentric representation~\cite{Pup17}.}
	\label{fig:orderedConcentric}
\end{figure}
Moreover, in an ordered concentric representation, the following three properties hold: 

\begin{enumerate}[(R.1)]

\item \label{r:level} Each level edge is drawn outside the circle on which its endvertices are located, 

\item \label{r:binding} Each binding edge consists of at most two segments, one required segment which is drawn between the two circles on which its endvertices are located and one optional segment which is drawn outside both of these circles, and

\item \label{r:root} Vertex $r$ is incident to the outer face.

\end{enumerate}

\noindent For an illustration refer to Figure~\ref{fig:orderedConcentric}. To compute an ordered concentric representation of a planar graph $G$, Pupyrev~\cite{Pup17} first computes a BFS-tree of $G$. The root of this tree is the center of the representation, while its edges are binding edges drawn between the two corresponding circles (see R.\ref{r:binding}). His result is summarized as follows.

\begin{lemma}[Pupyrev~\cite{Pup17}]
Given an $n$-vertex planar graph $G=(V,E)$, it is possible to compute in $\Oh(n)$ time an ordered concentric representation $\{\mathcal{C}_0,\mathcal{C}_1,\ldots,\mathcal{C}_{h-1}\}$ of $G$ with center $r \in V$, such that $h$ is the height of a breadth-first search tree of $G$ rooted at $r$.
\label{lem:pupyrev}
\end{lemma}

Given a spanning tree $T$ of a graph $G$ rooted at a vertex $r$, and a vertex $v$ of $G$, we denote by $\mbox{dist}(v)$ the graph-theoretic distance of $v$ from $r$ in $T$. Clearly, $\mbox{dist}(r)=0$, while for each leaf $v$ of $T$, it holds that $\mbox{dist}(v) \leq h(T)$, where $h(T)$ is the height of $T$. For a vertex $v$ of $G$, we refer to the value $h(T)-\mbox{dist}(v)$ as the \emph{layer} of $v$ in $T$, which we denote by $\ell(v)$. Note that given an ordered concentric representation $R=\{\mathcal{C}_0,\mathcal{C}_1,\ldots,\mathcal{C}_{h-1}\}$ of a graph $G$ computed using a BFS-tree $T$ of $G$, for each vertex that belongs to circle $\mathcal{C}_i$ in $R$, its layer in $T$ is $i$, for each $i=0,1,\ldots,h-1$.

\section{The Special Case: \fullBM Graphs}
\label{sec:subcase}
In this section we consider a special subfamily of planar graphs of degree $\Delta$, defined as follows. 

\begin{definition}\label{def:deltamg}
A \emph{\fullBM} graph $G=(V,E)$  consists of a $(\Delta-1)$-ary tree $T$ with vertex set $V$, rooted at a vertex $r \in V$, and a perfect matching $M$ on the leaves of $T$, such that there exists a planar drawing $\Gamma$ of $G$ with the following properties (refer to Figure~\ref{fig:special-case} for an illustration):

\begin{enumerate}[({P.}1)]
\item \label{p:leaf} all leaves of $T$ lie on a horizontal line $L:y=0$,
\item \label{p:mt} all edges of $T$ (of $M$) are drawn above (below, respectively) $L$ 
\item \label{p:r} vertex $r$ lies on the outer face of $\Gamma$, and 
\item \label{p:layers} the vertices of layer $\ell>0$ in $T$ are drawn on a horizontal line $L_\ell:y=\ell$.
\end{enumerate}
\end{definition}

Note that the root $r$ of $T$ has degree at most $\Delta-1$, while the leaves of $T$ have degree~$2$ in $G$. All other vertices of $G$ have degree at most $\Delta$. In the particular example depicted in Figure~\ref{fig:special-case} each internal vertex has degree exactly $\Delta$.  Hence, graph $G$ has maximum degree at most $\Delta$.  Our goal is to compute a queue layout of $G$  with at most $2\Delta-2$ queues. In this direction, it is worth observing that the treewidth of \fullBM graphs is not bounded; in particular, a suitable subdivision of a grid graph yields a \fullBM graph. Observe that drawing $\Gamma$ can be easily converted into an ordered concentric representation of $G$ with center $r$ in which all binding edges are part of $T$. 

First, we assign an integer value, called \emph{nesting-value}, to each edge of $M$. Recall that each edge of $M$ connects two vertices of $G$ that are leaves in $T$, and by Property~P.\ref{p:leaf} of Definition~\ref{def:deltamg} are along the horizontal line $L:y=0$ in the drawing $\Gamma$ of $G$. We denote the order in which the endvertices of the edges of $M$ appear along $L$ by $\prec_L$. Consider now an edge $e$ of $M$. Edge $e$ has nesting-value zero if there does not exist an edge $e'$ that nests $e$ with respect to $\prec_L$. Edge $e$ has nesting-value $i>0$ if the maximum nesting-value of all edges nesting $e$ is equal to $i-1$; see Figure~\ref{fig:special-case} for an illustration.

\begin{figure}
	\centering
	\includegraphics[width=0.7\textwidth]{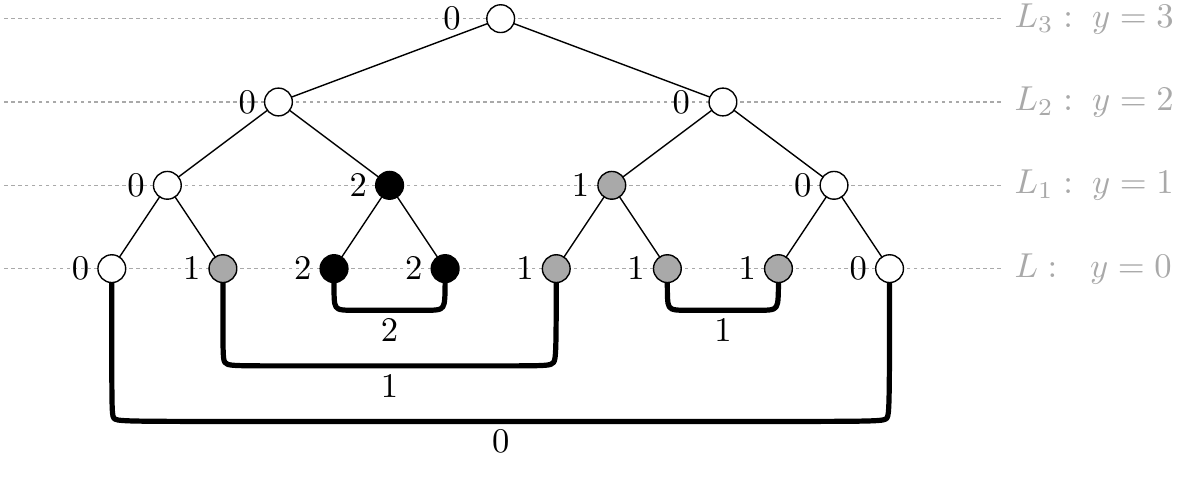}
	\caption{Illustration of a \fullBM graph $G$ with $\Delta=3$, in which
	the edges of $M$ are drawn bold below the horizontal line $L:\;y=0$.
	The label of each edge of $M$ corresponds to its nesting-value, 
	while the label of each vertex of $G$ corresponds to its matching-value. 
	Note that every internal vertex of the tree $T$ of $G$ has degree exactly $\Delta$ in this example, even though in general internal vertices of $T$ may have degree smaller than $\Delta$ (but greater than one). }
	\label{fig:special-case}
\end{figure}

Based on the nesting-values of the edges of $M$, we compute in a bottom-up traversal of $T$ an integer value for each vertex $v$ of $G$, the so-called \emph{matching-value} of $v$, $\mbox{mv}(v)$. For a vertex $v$ on layer $0$ of $T$, the matching-value of $v$ is the nesting-value of the unique edge of $M$ incident to it. For a vertex $v$ on layer $\ell \geq 1$ of $T$, the matching-value of $v$ equals the minimum matching-value of its neighbors in layer $\ell-1$ in $T$; refer to Figure~\ref{fig:special-case} for an illustration. In other words, the matching-value of vertex $v$ equals the minimum nesting-value of the edges of $M$ incident to the leaves of the subtree of $T$ rooted at $v$. In addition, the matching-values of any two consecutive leaves of $T$ along $L$ differ by at most one. Hence, if the leaves of a subtree rooted at a vertex $v$ have minimum matching-value $\alpha$ and maximum matching-value $\beta$, then for every value $m$ in $[\alpha,\beta]$ there exists at least one leaf of this subtree with matching-value $m$.

Finally, based on the matching-values, we partition the vertices of $G$ that belong to a certain layer of $T$ to \emph{layer-groups}. Formally, the layer-group $\mbox{g}(v)$ of a vertex $v$ of $G$ is defined as follows, where $(\Delta-1)^{\ell(v)}$ is the maximum number of leaves in the subtree of $T$ rooted at $v$:
\begin{equation}
\mbox{g}(v):= \left\lfloor \frac{\mbox{mv}(v)}{(\Delta-1)^{\ell(v)}} \right\rfloor.
\label{eq:layergroup}
\end{equation}
Observe that $\mbox{mv}(v) \in [ \mbox{g}(v)\cdot (\Delta-1)^\ell, (\mbox{g}(v)+1)\cdot (\Delta-1)^\ell )$; in particular, if $v$ belongs to layer $0$, $\mbox{mv}(v) = \mbox{g}(v)$. We further denote by $V_\ell^g$ the set of vertices of $G$ that belong to layer $\ell$ in $T$, and that are contained in layer-group~$g$. Remark that $\{V_\ell^g\}_{\ell,g}$ is a partition of the vertices of $G$. 

We are now ready to present the main result of this section.

\begin{lemma}
Every \fullBM graph has queue number at most $2\Delta-2$.
\label{lem:fullBM}
\end{lemma}
\begin{proof}
Let $G$ be a \fullBM graph, and let $\{V_\ell^g\}_{\ell,g}$ be a partition of the vertex-set of $G$ as described above. 
Recall that $\Gamma$ is a planar drawing of $G$ satisfying Properties~P.\ref{p:leaf}--P.\ref{p:layers} of Definition~\ref{def:deltamg}.
We construct a linear order $\prec$ of the vertices of $G$ as follows. For every two distinct vertices $u$ and $v$ of $G$, we have that $u \prec v$ if and only if one of the following conditions holds:
\begin{enumerate}[(C.1)]
\item\label{c1:linear} $\ell(u)<\ell(v)$, or 
\item\label{c2:linear} $\ell(u)=\ell(v)$ and $\mbox{g}(u)<\mbox{g}(v)$, or
\item\label{c3:linear} $\ell(u)=\ell(v)=\ell$, $\mbox{g}(u)=\mbox{g}(v)$ and $u$ is to the left of $v$ along $L_\ell$ in $\Gamma$.
\end{enumerate}

Observe that all edges of $M$ form a necklace in $\prec$. Indeed, if $e$ is an edge of $M$, then the endvertices of $e$ belong to the same layer-group. Further, since $e$ only nests matching edges of larger nesting value with respect to $\prec_L$, its endpoints are consecutive in $\prec$. Therefore, we can assign all edges of $M$ in one queue, say $\mathcal{Q}_0$. 

The remaining edges of $G$ belong to $T$. As a result, their endvertices belong to consecutive layers of $T$. Let $u$ be a vertex of $G$ that belongs to layer~$\ell+1$ in $T$, and assume that $u$ is in layer-group~$g$, i.e., $u\in V_{\ell+1}^g$. Let also $v_1,v_2,\dots,v_d$ be the neighbors of $u$ in layer $\ell$ (not necessarily in this left-to-right order), where $d \le \Delta-1$. Without loss of generality, we assume that $\mbox{mv}(v_1)\leq \mbox{mv}(v_2)\leq \dots \leq \mbox{mv}(v_{d})$. This implies that $\mbox{mv}(u)=\min\{\mbox{mv}(v_i) \mid i=1,\ldots,d\} =\mbox{mv}(v_1)$. Since $u$ is in $V_{\ell+1}^g$, it follows by (\ref{eq:layergroup}) that 
\begin{equation}
\mbox{mv}(v_1) = \mbox{mv}(u) \in \left[g\cdot (\Delta-1)^{\ell+1},\;(g+1)\cdot (\Delta-1)^{\ell+1}\right).
\label{eq:mv_u}
\end{equation}
Hence
\begin{equation}
\begin{split}
g\cdot(\Delta-1)\cdot (\Delta-1)^\ell& = g\cdot (\Delta-1)^{\ell+1}\\ 
~ &\leq \mbox{mv}(v_1)\\ 
~ &< (g+1)\cdot (\Delta-1)^{\ell+1}\\
~ &= (g+1)\cdot(\Delta-1) \cdot (\Delta-1)^\ell.
\end{split}
\label{eq:mv-value-v1}
\end{equation} 
It follows that for vertex $v_1$ of layer $\ell$ we have
\[
 g\cdot(\Delta-1) \leq \mbox{g}(v_1) < (g+1)\cdot(\Delta-1).
\]
Alternatively,
\[
 v_1\in V_\ell^{g\cdot(\Delta-1)}\cup \dots \cup V_\ell^{(g+1)\cdot(\Delta-1)-1} = \bigcup_{k=0}^{\Delta-2}V_\ell^{g\cdot(\Delta-1)+k}.
\]

\begin{figure*}[t]
	\centering
	\includegraphics[width=\textwidth,page=2]{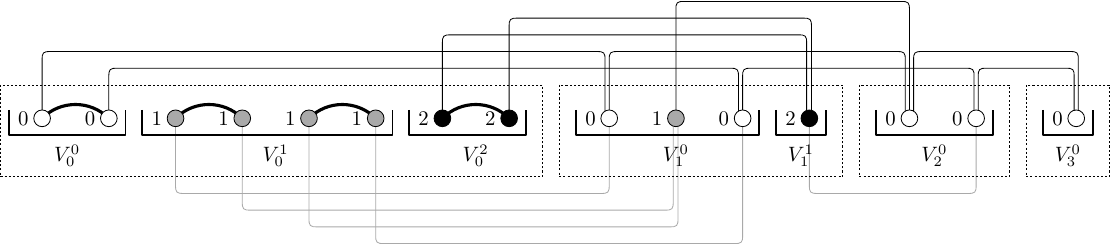}
    \caption{%
    The queue layout for the \fullBM graph $G$ of Figure~\ref{fig:special-case} produced by our algorithm. 
    This layout consists of four queues drawn black-bold, black-solid and gray (denoted by $\mathcal{Q}_0$, $\mathcal{Q}_1$ and $\mathcal{Q}_2$ in the algorithm).
    The vertices of $G$ that belong to the same layer in $T$ are drawn within the same dotted rectangle. }
	\label{fig:queue-layout}
\end{figure*}

\noindent Now consider vertex $v_i$, such that $1 < i \leq d$. 
We claim that
\begin{equation}
g\cdot(\Delta-1) \leq \mbox{g}(v_i) < (i-1) + (g+1)\cdot(\Delta-1)
\label{eq:claim-for-vi}
\end{equation}

\noindent Recall that the matching-values of two consecutive leaves differ by at most 1, and that $\mbox{mv}(v_1)\leq \mbox{mv}(v_2)\leq \dots \leq \mbox{mv}(v_{d})$. It follows that  $\mbox{mv}(v_i)$ and $\mbox{mv}(v_{i-1})$ differ by at most the number of leaves in the subtree rooted at $v_{i-1}$, that is 
\begin{equation}
\mbox{mv}(v_i)\leq (\Delta-1)^\ell + \mbox{mv}(v_{i-1})\notag,
\end{equation}
\noindent which applied $i-1$ times, gives
\begin{equation}
\mbox{mv}(v_i) \leq (i-1)(\Delta-1)^\ell +\mbox{mv}(v_1).
\label{eq:mv_v}
\end{equation}
\noindent \eqref{eq:mv-value-v1} and~\eqref{eq:mv_v} give
\begin{equation}\label{eq:mv_vi}
\begin{split}
\mbox{mv}(v_i) & < (i-1)(\Delta-1)^\ell + (g+1)\cdot(\Delta-1)\cdot(\Delta-1)^\ell\\
~ & = \big((i-1) + (g+1)\cdot(\Delta-1)\big)\cdot(\Delta-1)^\ell.\notag
\end{split}
\end{equation}
\noindent Since $\mbox{mv}(v_1) \leq \mbox{mv}(v_i)$, we conclude that
\begin{equation}
 g\cdot(\Delta-1)\cdot (\Delta-1)^\ell \leq \mbox{mv}(v_i) < \big((i-1) + (g+1)\cdot(\Delta-1)\big)\cdot(\Delta-1)^\ell\notag,
\end{equation}
which implies our initial claim in~\eqref{eq:claim-for-vi}, as desired. Using~\eqref{eq:claim-for-vi} and the fact that $i \leq d \leq \Delta-1$, we get for $i = 1,\ldots, d$,
\begin{equation}
 v_i \in V_\ell^{g\cdot(\Delta-1)}\cup \dots \cup V_\ell^{(g+2)\cdot(\Delta-1)-2} = \bigcup_{k=0}^{2\Delta-4}V_\ell^{g\cdot(\Delta-1)+k}.\label{eq:summary}
\end{equation}

Recall that we have already assigned all edges of $G$ that belong to $M$ in a single queue denoted by $\mathcal{Q}_0$. 
We are now ready to describe how to assign each edge of $G$ that is in $T$ to the remaining $2\Delta - 3$ queues. For each layer $\ell>0$ and for each layer-group $g$ of layer $\ell$ in $T$, we assign the edges between $V_{\ell+1}^g$ and $V_\ell^{g\cdot(\Delta-1)+k}$ to queue $\mathcal{Q}_{k+1}$, for $k=0,\ldots, 2\Delta-4$; for an illustration refer to Figure~\ref{fig:queue-layout}. By~(\ref{eq:summary}), all edges of $G$ have been assigned to one of the $2\Delta - 2$ queues $\mathcal{Q}_0,\ldots,\mathcal{Q}_{2\Delta-3}$. 

To complete the proof of the lemma, it remains to show that no two independent edges of the same queue are nested. As already mentioned, all edges of $\mathcal{Q}_0$ form a $|\mathcal{Q}_0|$-necklace and therefore do not nest. For some $0 < k \leq 2\Delta-3$, consider two independent edges $(x_1,y_1)$ and $(x_2,y_2)$ assigned to the same queue $\mathcal{Q}_{k}$, such that $x_1 \prec y_1$ and $x_2 \prec y_2$. Without loss of generality, we may further assume that $x_1 \prec x_2$. We will prove that $y_1 \prec y_2$, which implies that edges $(x_1,y_1)$ and $(x_2,y_2)$ are not nested. 

For the sake of contradiction, assume that $x_1 \prec x_2 \prec y_2 \prec y_1$, that is, $(x_1,y_1)$ nests $(x_2,y_2)$ with respect to $\prec$. Since $x_1 \prec x_2$, by Conditions~C.\ref{c1:linear}-C.\ref{c3:linear}, we have that $\ell(x_1)\leq\ell(x_2)$. Since the endvertices of each edge in $T$ belong to consecutive layers and since $(x_1,y_1)$ and $(x_2,y_2)$ belong to $T$, it follows that $\ell(y_1)=\ell(x_1)+1$ and $\ell(y_2)=\ell(x_2)+1$. Since $y_2 \prec y_1$, it follows that $\ell(y_2)\leq\ell(y_1)$ and hence $\ell(x_2)+1 \leq \ell(x_1)+1$. Since $\ell(x_1)\leq\ell(x_2)$, we may conclude that $\ell(x_1)=\ell(x_2)=\ell$ and $\ell(y_1)=\ell(y_2)=\ell+1$, for some $\ell\geq0$.

Let $\mbox{g}(y_1) = g_1$ and $\mbox{g}(y_2) = g_2$. Since both $(x_1,y_1)$ and $(x_2,y_2)$ are assigned to queue $\mathcal{Q}_{k}$, it follows that $\mbox{g}(x_1) = g_1\cdot(\Delta-1)+(k-1)$ and $\mbox{g}(x_2) = g_2\cdot(\Delta-1)+(k-1)$. Since $x_1 \prec x_2$, by Conditions~C.\ref{c2:linear} and~C.\ref{c3:linear}, we conclude that $g_1\cdot(\Delta-1)+(k-1) = \mbox{g}(x_1) \leq \mbox{g}(x_2) = g_2\cdot(\Delta-1)+(k-1)$, which implies that $g_1\leq g_2$. On the other hand, since $y_2 \prec y_1$, by Conditions~C.\ref{c2:linear} and~C.\ref{c3:linear} we similarly get $g_2\leq g_1$. Therefore, $g_1=g_2$ and consequently $\mbox{g}(x_1) = \mbox{g}(x_2)$.

Now since $x_1 \prec x_2$, by Condition~C.\ref{c3:linear} it follows that $x_1$ is to the left of $x_2$ along $L_\ell$ in the drawing $\Gamma$ of $G$. Similarly, since $\mbox{g}(y_1) = \mbox{g}(y_2)$ and $y_2 \prec y_1$, it follows by Condition~C.\ref{c3:linear} again that $y_2$ is to the left of $y_1$ along $L_{\ell+1}$ in the drawing $\Gamma$ of $G$. However, this implies that edges $(x_1,y_1)$ and $(x_2,y_2)$ cross in $\Gamma$, which is a contradiction to the fact that $\Gamma$ is a planar drawing of $G$. Therefore, $y_1 \prec y_2$, and edges $(x_1,y_1)$ and $(x_2,y_2)$ are not nested, as desired, which concludes the proof.
\end{proof}

\section{The Generalization: Planar Graphs of Maximum Degree~$\Delta$}
\label{sec:general}
In this section, we use the approach presented in the previous section to general planar graphs of maximum degree $\Delta$. At a high-level, our approach consists of three main steps. Given a planar graph $G$ of maximum degree $\Delta$, we first compute an auxiliary planar graph $G_1$ of maximum degree $\Delta$ by subdividing some of the edges of $G$ a constant number of times. Then, we exploit structural properties of graph $G_1$ to obtain a \fullBM graph $G_2$ of maximum degree $\Delta$ by replacing some of the vertices of $G_1$ with appropriately-defined \fullBM instances. It follows, by Lemma~\ref{lem:fullBM}, that the queue number of $G_2$ is at most $2\Delta-2$. In a third step, we show that a queue layout of $G$ can be obtained from a queue layout of $G_2$ by introducing a number of additional queues that is polynomial in $\Delta$, thus proving Theorem~\ref{thm:main}.

\begin{figure*}
\centering
	\subcaptionbox{\centering \label{fig:subdivisions-1}}{
	\includegraphics[width=0.35\textwidth,page=3]{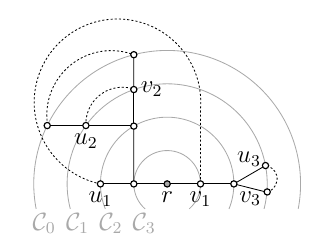}}
	\hfil
	\subcaptionbox{\centering \label{fig:subdivisions-2}}{
	\includegraphics[width=0.35\textwidth,page=4]{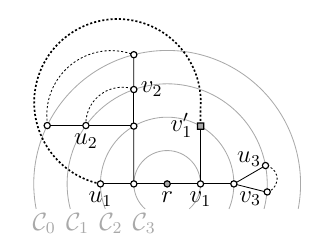}}
	\subcaptionbox{\centering \label{fig:subdivisions-3}}{
	\includegraphics[width=0.35\textwidth,page=5]{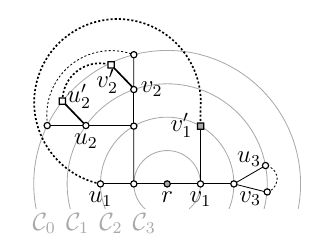}}
    \hfil
	\subcaptionbox{\centering \label{fig:subdivisions-4}}{
	\includegraphics[width=0.35\textwidth,page=6]{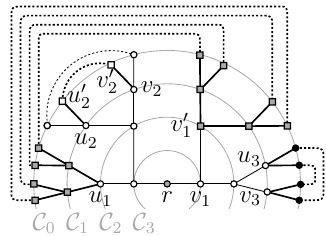}}
    \caption{%
    Illustration of different operations performed by our algorithm. 
    The starting configuration is the ordered concentric representation of a planar graph of degree $3$ illustrated in~(a), in which     
    $(u_1,v_1)$ is a binding edge, while
    $(u_2,v_2)$ and $(u_3,v_3)$ are level edges.}
	\label{fig:subdivisions}
\end{figure*}

 \subsection{Step 1: Construction of the auxiliary planar graph $\mathbf{G_1}$}\label{step:g1} 
 First, let us argue that we may assume without loss of generality that $G$ has minimum degree at least~two. Note that Theorem~\ref{thm:main} clearly holds for $\Delta \leq 2$, as all graphs of maximum degree at most~two have queue number one~\cite{HR92}. So, assume for the remainder that $\Delta \geq 3$. Suppose that $v$ is a vertex of degree~one in $G$. We introduce two new vertices $v_1$ and $v_2$, and three new edges $(v,v_1)$, $(v,v_2)$ and $(v_1,v_2)$ in $G$. It follows that vertex $v$ has degree~three and each of the two introduced vertices $v_1$ and $v_2$ has degree~two. By applying this procedure to all degree-$1$ vertices in $G$, we obtain a planar supergraph $G'$ of $G$ with minimum degree at least~two and maximum degree $\Delta$. Since $G$ is a subgraph of $G'$ the queue number of $G$ is at most the queue number of $G'$. So, in the following we will assume that every vertex in $G$ has degree at least~two.

We are now ready to describe how to construct graph $G_1$ from graph $G$. If $G$ has no degree-$2$ vertex, then we subdivide an edge of $G$ once, in order to introduce such a degree-$2$ vertex. As a consequence, in the following we will assume that $G$ contains at least one degree-$2$ vertex, which we denote by~$r$. Let $R=\{\mathcal{C}_0,\mathcal{C}_1,\ldots,\mathcal{C}_{h-1}\}$ be an ordered concentric representation of $G$ centered at $r$. Let $T$ be the BFS-tree of $G$ that was used in order to compute $R$; refer to Section~\ref{sec:preliminaries}. Recall that an edge of $G$ that is not in~$T$, is either a level edge (if it connects two vertices of the same level in $T$) or a binding edge (if it connects two vertices on consecutive levels in $T$).

We proceed by subdividing each binding edge of $G$ that does not belong to $T$ (if any) once. By this operation, each binding edge of $G$ is split into a level edge and an edge that can be assigned to tree $T$. For an illustration of this operation refer to the binding edge $(u_1,v_1)$ in Figure~\ref{fig:subdivisions-1}, which is subdivided once by introducing vertex $v_1'$ in Figure~\ref{fig:subdivisions-2}; as a result, the edge $(v_1,v_1')$ is binding and part of $T$, while the edge $(u_1,v_1')$ is level. Additionally, we subdivide each level edge $(u,v)$ twice, if $u$ or $v$ has degree greater than two. This yields three edges $(u,u')$, $(u',v')$ and $(v',v)$, the first and last of which we assign to the tree $T$, while the middle edge $(u',v')$ becomes a level edge. This guarantees that the endvertices of all level edges have degree two. For an illustration refer to the level edge $(u_2,v_2)$ in Figure~\ref{fig:subdivisions-2} whose endvertices have degree three; by subdividing this edge twice with vertices $u_2'$ and $v_2'$ in Figure~\ref{fig:subdivisions-3}, we guarantee that two of the newly formed edges, that is, $(u_2,u_2')$ and $(v_2,v_2')$, become binding and part of $T$, while the middle edge $(u_2',v_2')$ becomes a level edge whose endvertices have degree two. The resulting subdivision of $G$ is graph $G_1$. Note that each edge of $G$ is subdivided at most three times in order to obtain $G_1$. 

\subsection{Step 2: Construction of the \fullBM graph $\mathbf{G_2}$}\label{step:g2} Based on graph $G_1$, we now describe how to construct graph $G_2$. In this step, we need to guarantee one additional property of \fullBM graphs, that is, all level edges of a \fullBM graph belong to layer~$0$. Suppose that in the concentric representation of $G_1$ there exists a level edge $(u,v)$ at layer $\ell$, where $0<\ell<h$. We create two complete $(\Delta-1)$-ary trees of height $\ell$, say $T_u$ and $T_v$, and we identify their roots with vertices $u$ and $v$, respectively. Since the height of each of $T_u$ and $T_v$ is $\ell$, it follows that all leaves of $T_u$ and all leaves of $T_v$ can be placed along $\mathcal{C}_0$ consecutively while preserving planarity. We replace edge $(u,v)$ by $(\Delta-1)^{\ell}$ edges forming a matching $M_{u,v}$, such that the $i$-th leaf of $T_u$ from the left along $\mathcal{C}_0$ is connected to the $i$-th leaf of $T_v$ from the right along $\mathcal{C}_0$. Observe that the edges of $M_{u,v}$ form a $(\Delta-1)^{\ell}$-rainbow in $\prec_L$. For an illustration of this operation, refer to edges $(u_1,v_1')$ and $(u_3,v_3)$ in Figure~\ref{fig:subdivisions-3}, which are two level edges that do not have their endvertices along the outermost circle $\mathcal{C}_0$. Edge $(u_1,v_1')$ is replaced by the two gray-colored binary trees $T_{u_1}$ and $T_{v_1'}$ of height two in Figure~\ref{fig:subdivisions-4}, while the edge $(u_3,v_3)$ is replaced by the two black-colored binary trees $T_{u_3}$ and $T_{v_3}$ of height one. Also, the four edges of $M_{u_1,v_1'}$ and the two edges of $M_{u_3,v_3}$ are drawn dotted with mostly rectilinear segments in Figure~\ref{fig:subdivisions-4}. Once we apply the aforementioned procedure to all level edges of $G_1$, we obtain graph $G_2$ together with an ordered concentric representation $R_2$, in which all level edges have their endvertices along the outermost circle $\mathcal{C}_0$. Since we assumed that $G$ has minimum degree at least~two, all leaves of $T$ lie on $\mathcal{C}_0$.

We now claim that $R_2$ can be converted into a drawing $\Gamma$ of $G_2$ that satisfies Properties P.\ref{p:leaf}--P.\ref{p:layers} of Definition~\ref{def:deltamg}. Since $r$ belongs to the outer face of $R_2$, there is a curve $\mathcal C$ that starts at $r$, cuts all circles of $R_2$ once, and crosses no edge of $G_2$. Hence, we can use $\mathcal C$ to cut all the circles of $R_2$ and stretch $R_2$ so that each circle becomes a line segment by preserving the planarity of the drawing. In other words, we can obtain the drawing $\Gamma$ of $G_2$ through a suitable homeomorphic transformation of $R_2$. Properties P.\ref{p:leaf} and~P.\ref{p:layers} hold by construction in $\Gamma$, while Properties P.\ref{p:mt} and~P.\ref{p:r} follow from Properties R.\ref{r:level} and~R.\ref{r:root} of the ordered concentric representation $R_2$, respectively. It follows that graph $G_2$ is a \fullBM graph. 

\subsection{Step 3: Construction of a queue layout of $\mathbf{G}$}\label{step:queue} In this step we describe how to construct a queue layout of $G$ with $O(\Delta^6)$ queues by exploiting structural properties of $G_1$ and $G_2$. First we observe that since $G_2$ is a \fullBM graph, by Lemma~\ref{lem:fullBM} it admits a queue layout $\mathcal{L}_2$ with $2\Delta-2$ queues $\mathcal{Q}_0$, $\mathcal{Q}_1$, $\ldots$, $\mathcal{Q}_{2\Delta-3}$. Also, graph $G_2$ contains a subdivision of $G$ as an induced subgraph. 

Next we derive a queue layout $\mathcal{L}_1$ for $G_1$ from the queue layout $\mathcal{L}_2$ for $G_2$. Recall that $G_2$ was obtained from $G_1$ by replacing each level edge $(u,v)$ of $G_1$ by two complete $(\Delta-1)$-ary trees $T_u$ and $T_v$, and by a set $M_{u,v}$ of matching edges. For the desired queue layout $\mathcal{L}_1$ of $G_1$ we order the vertices of $G_1$ according to their ordering in $\mathcal{L}_2$. For every edge of $G_1$ we assign the same queue as in $\mathcal{L}_2$, provided that this edge is also an edge of $G_2$. Otherwise, such an edge is a level edge of $G_1$ and we assign it to the queue $\mathcal{Q}_0$. Thus, in queue layout $\mathcal{L}_1$, all level edges are assigned to queue $\mathcal{Q}_0$, while queues $\mathcal{Q}_1,\ldots,\mathcal{Q}_{2\Delta-3}$ contain all edges of the BFS-tree $T$.

It remains to show that no two (level) edges in $\mathcal{Q}_0$ nest. Recall that the ordering of vertices is inherited from queue layout $\mathcal{L}_2$ and hence satisfies Conditions C.\ref{c1:linear}--C.\ref{c3:linear}. For any level $(u,v)$ in $G_1$ we have $\ell(u) = \ell(v)$. Moreover, the set of matching edges in $G_2$ incident to the leaves of $T_u$ is given by $M_{u,v}$, and the same holds for $T_v$. This implies that $\mbox{mv}(u)=\mbox{mv}(v) = \min\{ \text{nesting-value of $e$} \mid e \in M_{u,v}\}$ and consequently $\mbox{g}(u) = \mbox{g}(v)$ by~\eqref{eq:layergroup}. Now suppose for the sake of contradiction that level edge $(u_1,v_1)$ nests level edge $(u_2,v_2)$ with $u_1 \prec u_2 \prec v_2 \prec v_1$. By Condition~C.\ref{c1:linear} it follows that $\ell(u_1) \leq \ell(u_2) \leq \ell(v_2) \leq \ell(v_1)$ and as $\ell(u_1) = \ell(v_1)$, all four vertices have the same level $\ell$. Secondly, in $G_2$ all edges of $M_{u_1,v_1}$ nest all edges of $M_{u_2,v_2}$. Hence, $\mbox{mv}(v_2) - \mbox{mv}(v_1) = |M_{u_2,v_2}| = (\Delta-1)^{\ell}$, i.e., the difference between the minimum nesting-value of edges in $M_{u_2,v_2}$ and the minimum nesting-value of edges in $M_{u_1,v_1}$ is $(\Delta-1)^{\ell}$. This however, by~\eqref{eq:layergroup} gives $\mbox{g}(v_2) = \mbox{g}(v_1) + 1$, which contradicts the fact that $v_2 \prec v_1$ according to Condition~C.\ref{c2:linear}.

So far, we constructed a queue layout $\mathcal{L}_1$ of graph $G_1$ with $2\Delta-2$ queues. As $G_1$ is obtained from $G$ by subdividing each edge at most three times, we can apply Lemma~\ref{lem:subdivision} with $k=3$ and $q=2\Delta-2$, and conclude that there is a queue layout $\mathcal{L}$ of $G$ with at most $32(2\Delta-1)^{6}-1$ queues. Hence, $\Qn(G) \leq 32(2\Delta-1)^{6}-1$, which concludes the proof of Theorem~\ref{thm:main}.

\section{Time Complexity}
\label{sec:timeComplexity}

In this section, we analyze the time complexity of our algorithm to construct a queue layout with at most $32(2\Delta-1)^6-1$ queues for a given planar graph $G$ of maximum degree $\Delta$. First assume that the input graph is a \fullBM graph, we prove the following lemma.

\begin{lemma}\label{lem:fullBMtime}
Every \fullBM graph $G=(V,E)$ of maximum degree $\Delta$ admits a queue layout with at most $2\Delta-2$ queues, which can be computed in $\Oh(|V|)$ time.
\end{lemma}
\begin{proof}
Let $G=(V,E)$ be a \fullBM graph of maximum degree $\Delta$. By Lemma~\ref{lem:fullBM}, graph $G$ admits a queue layout with at most $2\Delta-2$ queues. The time complexity of the algorithm used to construct such a queue layout is dominated by the computation of the nesting-values of the edges of $M$ and of the matching-values of the vertices of $G$. The former can be easily accomplished in $\Oh(|M|)$ time based on the drawing $\Gamma$ of $G$, while the latter in $\Oh(|V|)$ by a bottom-up traversal of $T$. Having computed these values, the calculation of the layer-group of each vertex can be done in $\Oh(|V|)$ time in total. The linear order can be determined in $\Oh(|V|)$ time by a single iteration through the vertices of each layer $\ell$ of $T$ in the order that they appear along $L_\ell$; recall Conditions~C.\ref{c1:linear}--C.\ref{c3:linear}. Finally, the assignment of the edges of $G$ into queues can be performed in $\Oh(|E|)$ time, since for each edge the corresponding queue is determined based on the layer-groups of its endvertices. Hence, the algorithm supporting Lemma~\ref{lem:fullBM} runs in $\Oh(|V|+|E|)$ time, which is in $\Oh(|V|)$, since $G$ is planar and hence $|E| = \Oh(|V|)$.
\end{proof}

\begin{figure}[t]
\centering
	\subcaptionbox{\centering \label{fig:path-1}}{
	\includegraphics[width=0.31\textwidth,page=7]{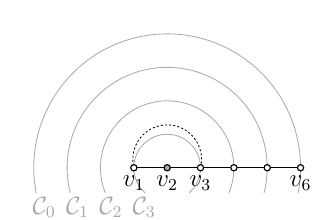}}
	\subcaptionbox{\centering \label{fig:path-2}}{
	\includegraphics[width=0.31\textwidth,page=8]{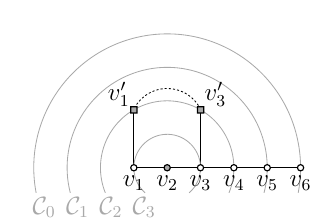}}
	\subcaptionbox{\centering \label{fig:path-3}}{
	\includegraphics[width=0.31\textwidth,page=9]{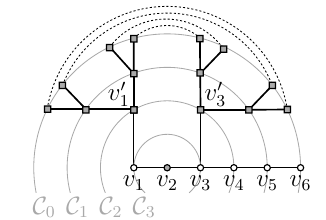}}
    \caption{%
    Illustration of the augmentation steps of our algorithm on a graph consisting of a path $\langle v_1,v_2,\ldots,v_6 \rangle$ and an edge $(v_1,v_3)$, assuming that the center of the ordered concentric representation is vertex $v_2$.}
	\label{fig:path}
\end{figure}

For the case of a general planar graph $G=(V,E)$ of maximum degree $\Delta$, the construction of the auxiliary graph $G_2$ may increase the time complexity dramatically. In particular, note that graph $G_2$ may have size exponential in $|V|$, when $G$ is sparse. To see this, consider a graph obtained from a path on $n$ vertices $\langle v_1,v_2,\ldots,v_n \rangle$ by adding an edge connecting $v_1$ to $v_3$. If $v_2$ is chosen as the center of the ordered concentric representation $R$, then edge $(v_1,v_3)$ becomes a level edge; see Figure~\ref{fig:path-1}. Since the endvertex $v_3$ of $(v_1,v_3)$ has degree greater than two, the edge $(v_1,v_3)$ will be subdivided twice. Let $v_1'$ and $v_3'$ be the subdivision vertices; refer to Figure~\ref{fig:path-2} for an illustration. Observe that the edge $(v_1',v_3')$ is a new level edge, whose endvertices are not in the outermost circle of the representation. Then each of the trees $T_{v_1'}$ and $T_{v_3'}$ replacing $(v_1',v_3')$ has $2^{n-3}-1$ vertices; see Figure~\ref{fig:path-3}. So, in order to keep the time complexity of our algorithm polynomial in the size of $G$, we avoid introducing trees $T_{u}$ and $T_{v}$ explicitly for each level edge $(u,v)$. In fact, the introduction of these trees was convenient for proving the correctness of our approach. But, as we will argue below, to determine the correct queue layout we only need to know the size of the set $M_{u,v}$. 

\begin{theorem}\label{thm:time}
Every planar graph $G=(V,E)$ of maximum degree $\Delta$ admits a queue layout with at most $32(2\Delta-1)^{6}-1$ queues, which can be computed in $\Oh(|V|^2\log{\Delta})$ time.
\end{theorem}
\begin{proof}
Let $G=(V,E)$ be a planar graph of maximum degree $\Delta$, which admits a queue layout with at most $32(2\Delta-1)^{6}-1$ queues by Theorem~\ref{thm:main}. 

To construct graph $G_1$ in Step~1 of our algorithm (refer to Section~\ref{step:g1}), we first need to compute the ordered concentric representation $R$ of $G$, which by Lemma~\ref{lem:pupyrev} can be done in linear time. Based on $R$, the computation of graph $G_1$ needs an additional $\Oh(|E|)$ time. 

In Step~2 of our algorithm (refer to Section~\ref{step:g2}), where we construct graph $G_2$, we need to avoid introducing trees $T_{u}$ and $T_{v}$ for each level edge $(u,v)$ of $T$. To achieve this, we first reroute each level edge in $R$, so that one part of it lies outside the outermost circle $\mathcal{C}_0$ of $R$. This can be done without introducing crossings in $\Oh(|E|)$ time in total, by a single traversal over all level edges in $R$ starting from those level edges of the outermost circle and moving inwards in the representation $R$. This edge rerouting guarantees that for any two level edges $(u,v)$ and $(u',v')$, all edges of $M_{u,v}$ nest all edges of $M_{u',v'}$ if and only if the part of $(u,v)$ that lies outside $\mathcal{C}_0$ of $R$ nests the corresponding part of $(u',v')$ in $R$. Therefore, instead of introducing two trees $T_{u}$ and $T_{v}$ for each level edge $(u,v)$ of $T$, we assign a \emph{weight} $w(u,v)$ to the edge $(u,v)$ equal to $|M_{u,v}|$ and compute the nesting-values of the edges of $M$ and the matching-values of the vertices of $G_2$ based on the weights of these edges as follows; e.g., the weight of the level edge $(v_1',v_3')$ of Figure~\ref{fig:path-2} is four. Consider an edge $(u,v)$ such that $(u,v)$ either belongs to $M$ or $(u,v)$ is a level edge of some layer $\ell>0$; note that in the former case we assume that $w(u,v)=1$. Observe that the order in which the endvertices of the edges of $M$ and the edge segments of the level edges lying outside $\mathcal{C}_0$  appear along $\mathcal{C}_0$ defines a linear order on their endvertices. Edge $(u,v)$ has nesting-value zero if there does not exist an edge $(u',v')$ that nests $(u,v)$ with respect to this order. Otherwise, let $(u',v')$ be the edge with maximum nesting-value that nests edge $(u,v)$. Then, edge $(u,v)$ has nesting-value equal to the nesting-value of $(u',v')$ plus $w(u',v')$. Note that once all nesting-values of edges are computed, the computation of the matching-values of vertices can be done as in the unweighted case. 
Up to this point, the time complexity of the algorithm is in $\Oh(|V|+|E|)=\Oh(|V|)$, assuming each operation can be performed in constant time. However, we operate with numbers that are not linear in $|V|$ (in particular, weights, nesting-values, and matching-values). On the other hand, these numbers do not exceed the maximum number of leaves of a $(\Delta-1)$-ary tree of height at most $|V|$, that is, $\Oh(\Delta^{|V|})$ or equivalently $\Oh(2^{|V|\log{\Delta}})$ . Also, we only need standard operations in our calculations (namely, additions and comparisons), hence we can perform each operation in $\Oh(|V|\log{\Delta})$ time, which increases the running time of the algorithm to $\Oh(|V|^2\log{\Delta})$.


Concerning Step 3 of our algorithm, we first note that we can construct a queue layout $\mathcal L_2$ of $G_2$ with $2\Delta-2$ queues in linear time by Lemma~\ref{lem:fullBMtime}. A queue layout $\mathcal L_1$ of $G_1$ with $2\Delta-2$ queues can be directly derived from $\mathcal L_2$ in linear time. The final queue layout $\mathcal L$ of $G$ is obtained by applying Lemma~\ref{lem:subdivision}. Although the time complexity of this lemma is not explicitly stated in~\cite[Lemma 27]{DW05}, one can easily observe that it is linear when the number of subdivision vertices is a constant (three in our case). In particular, the main operation of this lemma is a linear-time construction of a suitable track layout of a subgraph of $G$~\cite{DPW04}. This operation is repeated a number of times that is logarithmic in the number of subdivision vertices and hence constant in our case. This concludes the proof.
\end{proof}

\section{Open Problems}
\label{sec:conclusions}
As discussed in the introduction, our result prompted new developments on the long-standing conjecture by Heath, Leighton and Rosenberg. In particular, Dujmovi\'c et al.~\cite{DBLP:journals/corr/abs-1904-04791,focs} proved that the queue number of a planar graph is at most 49; clearly, improving this bound is a tempting question. However, as discussed by the authors in~\cite{DBLP:journals/corr/abs-1904-04791}, new ideas are probably required to obtain a significant improvement. Thus, we ask the following question, which may be easier to answer and of independent interest: Given a planar graph, is it possible to compute a linear order of its vertices and a partition of its edges into at most 3 stacks and fewer than 49 queues? These layouts are known as \emph{mixed layouts} and have been introduced in the paper by Heath, Leighton and Rosenberg, who conjectured that every planar graph admits a mixed layout with one stack and one queue. Pupyrev~\cite{Pup17} recently disproved this conjecture by demonstrating a planar graph for which one stack and one queue do not suffice, and conjectured that for bipartite planar graphs one stack and one queue are always sufficient.

\section*{Acknowledgements}
This work started at the Graphs and Network Visualisation workshop (GNV'18). We thank the organizers and the participants for fruitful discussions. We also thank David R.\ Wood for pointing out an issue in an earlier version of this paper and the anonymous referees of both the journal and the conference version of this paper for their valuable comments and suggestions.
\clearpage 
\bibliographystyle{abbrv}
\bibliography{paper}
\end{document}